\documentclass[11pt,letterpaper]{amsart}
\usepackage[leqno]{amsmath}
\usepackage{amsfonts,amsthm,amssymb}
\usepackage{mathrsfs}
\usepackage{epsfig}
\usepackage{graphicx}
\usepackage{subfigure}
\usepackage{url}
\usepackage{hyperref}
\usepackage[left=1.1in,right=1.1in,top=1.1in,bottom=1.1in,foot=0.6in]{geometry}
\usepackage{enumerate}
\usepackage{xspace}
\usepackage{xcolor}
\usepackage{cite}

\newtheorem{lemma}{Lemma}
\newtheorem{proposition}{Proposition}
\newtheorem{theorem}{Theorem}

\theoremstyle{definition}

\newtheorem{question}{Question}

\newtheorem{remark}{Remark}

\newtheorem{example}{Example}

\newcommand{\mcs}{\mathcal S}
\newcommand{\mcsij}{\mathcal{S}_{ij}}	 		

\DeclareMathOperator{\vcd}{VC-dim}		

\DeclareMathOperator{\vccdim}{cVC-dim}
\DeclareMathOperator{\vcsdim}{sVC-dim}
\newcommand{\vccdimt}{clique-VC-dimension\xspace}

\newcommand{\bouquet}{\mathcal B}
\newcommand{\Gbouquet}{G^*}
\newcommand{\Vbouquet}{V^*}
\newcommand{\Ebouquet}{E^*}

\newcommand{\incgraph}{G_{1,2}}
\newcommand{\phiij}{\varphi_{ij}}
\newcommand{\psiij}{\psi_{ij}}
\newcommand{\const}{\binom{d}{2}}

\begin{document}
	
	\thispagestyle{empty}
	
	\centerline{\Large\bf On density of subgraphs of halved cubes}
	
	\medskip
	\centerline{\it In memory of Michel Deza}

	\vspace{10mm}
	\centerline{Victor Chepoi, Arnaud Labourel, and S\'ebastien Ratel}

	\medskip
	\begin{small}
		\medskip
		\centerline{Laboratoire d'Informatique Fondamentale, Aix-Marseille Universit\'e and CNRS,}
		\centerline{Facult\'e des Sciences de Luminy, F-13288 Marseille Cedex 9, France}
		
		\centerline{\texttt{\{victor.chepoi, arnaud.labourel, sebastien.ratel\}@lif.univ-mrs.fr}}
	\end{small}
	
	\bigskip\bigskip\noindent
	{\footnotesize {\bf Abstract.}
		Let $\mcs$ be a family of subsets of a set $X$ of cardinality $m$ and $\vcd(\mathcal S)$ be the
		Vapnik-Chervonenkis dimension of $\mcs$. Haussler, Littlestone, and Warmuth (Inf. Comput., 1994)
		proved that if $G_1(\mcs)=(V,E)$ is the subgraph of the hypercube $Q_m$ induced by $\mcs$ (called the
		1-inclusion graph of $\mcs$), then $\frac{|E|}{|V|}\le \vcd({\mathcal S})$. Haussler (J. Combin. Th. A, 1995)
		presented an elegant proof of this inequality using the shifting operation.
		
		In this note, we adapt the shifting technique to prove that if $\mcs$ is an arbitrary set family and
		$\incgraph(\mcs)=(V,E)$ is the 1,2-inclusion graph of $\mcs$ (i.e., the subgraph of the square $Q^2_m$
		of the hypercube $Q_m$ induced by $\mcs$), then $\frac{|E|}{|V|}\le \binom{d}{2}$, where
		$d:=\vccdim^*(\mcs)$ is the clique-VC-dimension of $\mcs$ (which we introduce in this paper). The
		1,2-inclusion graphs are exactly the subgraphs of halved cubes and comprise subgraphs of Johnson
		graphs as a subclass.
	 }
	
	\section{Introduction} \label{introduction}
	Let $\mcs$ be a family of subsets of a set $X$ of cardinality $m$ and $\vcd(\mathcal S)$ be the
	Vapnik-Chervonenkis dimension of $\mcs$. Haussler, Littlestone, and Warmuth\cite[Lemma 2.4]{HaLiWa}
	proved that if $G_1(\mcs)=(V,E)$ is the subgraph of the hypercube $Q_m$ induced by $\mcs$ (called the {\it
		1-inclusion graph} of $\mcs$), then the following fundamental inequality holds: $\frac{|E|}{|V|} \le
	\vcd({\mathcal S})$. They used this inequality to bound the worst-case expected risk of a prediction model of
	learning of concept classes $\mcs$ based on the bounded degeneracy of their 1-inclusion graphs. Haussler
	\cite{Hau} presented an elegant proof of this inequality using the shifting (push-down) operation. 1-Inclusion
	graphs have many other applications in computational learning theory, for example, in sample compression
	schemes \cite{KuWa}. They are exactly the induced subgraphs of hypercubes and in graph theory they have
	been studied under the name of {\it cubical graphs} \cite{GaGr}.
	Finding a densest $n$-vertex subgraph of the hypercube $Q_m$ (i.e., an $n$-vertex subgraph $G$ of $Q_m$
	with the maximum number of edges) is equivalent to finding an $n$-vertex subgraph $G$ of $Q_m$ with the
	smallest edge-boundary (the number of edges of $Q_m$ running between $V$ and its complement in
	$Q_m$). This is the classical {\it edge-isoperimetric problem} for hypercubes \cite{Bezrukov, Har_book}.
	Harper \cite{Har}	nicely characterized the solutions of this problem: for any $n$, this is the subgraph of the
	hypercube induced by the initial segment of length $n$ of the {\it lexicographic numbering} of the vertices of
	the hypercube. One elegant way of proving this result is using compression \cite{Har_book}.
	
	Generalizing the density inequality $\frac{|E|}{|V|}\le \vcd({\mathcal S})$ of \cite{HaLiWa,Hau} to more general
	classes of graphs is an interesting and important problem. In the current paper, we present a density result
	for 1,2-inclusion graphs $\incgraph(\mcs)$ of arbitrary set families $\mcs$. The 1,2-inclusion graphs are the
	subgraphs of the square $Q_m^2$ of the hypercube $Q_m$ and they are exactly the subgraphs of the halved
	cube $\frac{1}{2}Q_{m+1}$ (Johnson graphs and their subgraphs constitute an important subclass). Since
	1,2-inclusion graphs may contain arbitrary large cliques for constant VC-dimension, we have to adapt the
	definition of classical VC-dimension to capture this phenomenon. For this purpose, we introduce the notion
	of {\it clique-VC-dimension} $\vccdim^*(\mcs)$ of any set family $\mcs$. Here is the main result of
	the paper:
	
	\begin{theorem}\label{density} Let $\mcs$ be an arbitrary set family of $2^X$ with $|X|=m$, let
		$d=\vccdim^*(\mcs)$ be the
		clique-VC-dimension of $\mcs$ and $\incgraph(\mcs)=(V,E)$ be the 1,2-inclusion graph of $\mcs$. Then
		$\frac{|E|}{|V|}\le \binom{d}{2}.$
	\end{theorem}
	
	\section{Related work} \label{sect_original_method}
	
	\subsection{Haussler's proof of the inequality $\frac{|E|}{|V|}\le \vcd({\mathcal S})$}
	We briefly review the notion of VC-dimension and the shifting method of \cite{Hau} of proving the inequality
	$\frac{|E|}{|V|}\le \vcd({\mathcal S})$ (the original proof of \cite{HaLiWa} was by induction on the number of
	sets). In the same vein, see Harper's proof \cite[Chapter 3]{Har_book} of the isoperimetric inequality via
	compression. We will use the shifting method in the proof of Theorem \ref{density}.
	
	Let $\mcs$ be a family of subsets of a set $X=\{ e_1,\ldots,e_m\}$; $\mcs$ can be viewed as a subset of
	vertices of the $m$-dimensional hypercube $Q_m$. Denote by $G_1(\mcs)$ the subgraph of $Q_m$ induced
	by the vertices of $Q_m$ corresponding to the sets of $\mcs$; $G_1(\mcs)$ is called the \emph{1-inclusion 	
		graph} of $\mcs$ \cite{Hau,HaLiWa}. Vice-versa, for any subgraph $G$ of $Q_m$ there exists a family of
	subsets $\mcs$ of $2^X$ such that $G$ is the 1-inclusion graph of $\mcs$. A subset $Y$ of $X$ is
	\emph{shattered} by $\mcs$ if for all $Y'\subseteq Y$ there exists $S\in\mcs$ such that $S\cap Y=Y'$. The
	\emph{Vapnik-Chervonenkis's dimension} \cite{VaCh} $\vcd(\mcs)$ of $\mcs$ is the cardinality of the largest
	subset of $X$ shattered by $\mcs$.
	
	\begin{theorem}[\negthickspace\cite{Hau,HaLiWa}] \label{subgraphs_hypercubes}
		If $G:=G_1(\mcs)=(V,E)$ is the 1-inclusion graph of a set family $\mcs\subseteq 2^X$ with VC-dimension
		$\vcd(\mcs)=d$, then $\frac{|E|}{|V|}\le d.$
	\end{theorem}
	
	For a set family ${\mathcal S}\subseteq 2^X,$ the {\it shifting (push down or stabilization) operation}
	$\varphi_e$ with respect to an element $e\in X$ replaces every set $S$ of $\mcs$ such that
	$S\setminus\{ e\}\notin {\mathcal S}$ by the set $S\setminus\{ e\}$. Denote by $\varphi_e({\mathcal S})$ the
	resulting set family and by
	$G'=G_1(\varphi_e(\mcs))=(V',E')$ the 1-inclusion graph of $\varphi_e({\mathcal S})$. Haussler \cite{Hau} 
    proved that the shifting map $\varphi_e$ has the following properties:
	
	\begin{itemize}
		\item[(1)] $\varphi_e$ is bijective on the vertex-sets: $|V|=|V'|$,
		\item[(2)] $\varphi_e$ is increasing the number of edges: $|E|\le |E'|$,
		\item[(3)] $\varphi_e$ is decreasing the VC-dimension: $\vcd(\mcs)\ge \vcd(\varphi_e({\mathcal S}))$.
	\end{itemize}
	Harper \cite[p.28]{Har_book} called {\it Steiner operations} the set-maps $\varphi: 2^X\rightarrow 2^X$
	satisfying (1), (2), and the
	following condition:
	\begin{itemize}
		\item[(4)] $S\subseteq T$ implies $\varphi (S)\subseteq \varphi (T)$.
	\end{itemize}
	He proved that the compression operation defined in \cite[Subsection
	3.3]{Har_book} is a Steiner operation. Note that $\varphi_e$ satisfies (4) (but is defined only on $\mcs$).
	
	After a finite sequence of shiftings, any set family $\mcs$ can be transformed into a set family $\mcs^*$,
	such that $\varphi_e(\mcs^*)=\mcs^*$ holds for any $e\in X$. 
	The resulting set family $\mcs^*$, a {\it complete shifting} of $\mathcal S$, is {\it downward
		closed} (i.e., is a {\it simplicial complex}). 
	Consequently, the 1-inclusion graph $G_1(\mcs^*)$ of $\mcs^*$ is a {\it bouquet of cubes}, i.e., a
	union of subcubes of $Q_m$ with a common origin $\varnothing$.
	Let $G^*=G_1(\mcs^*)=(V^*,E^*)$ and $d^*=\vcd(\mcs^*)$. Since all shiftings satisfy the conditions (1)-(3),
	we conclude that $|V^*|=|V|$, $|E^*|\ge |E|$, and $d^*\le d$.
	Therefore, to prove the inequality $\frac{|E|}{|V|}\le d$ it suffices to show that $\frac{|E^*|}{|V^*|}\le d^*$.
	Haussler deduced it from Sauer's lemma \cite{Sauer}, however it is easy to prove this inequality directly, by
	bounding the degeneracy of $G^*$. Indeed, let $v_0$ be the vertex of $G^*$ corresponding to the origin
	$\varnothing$ and let $v$ be a furthest from $v_0$ vertex of $G^*$. Then $v_0$ and $v$ span a maximal
	cube of $G^*$ (of dimension $\le d^*$) and $v$ belongs only to this maximal cube of $G^*$. Therefore, if we
	remove $v$ from $G^*$, we will also remove at most $d^*$ edges of $G^*$ and the resulting graph will be
	again a bouquet of cubes $G^-=(V^-,E^-)$ with one less vertex and dimension $\le d^*$. Therefore, we can
	apply the induction hypothesis to this bouquet $G^-$ and deduce that $|E^-|\le |V^-|d^*$. Consequently,
	$|E^*|\le d^*+ |E^-|\le d^*+(|V^*|-1)d^*=|V^*|d^*$.
	
	To extend Haussler's proof to subgraphs of halved cubes (and, equivalently, to subgraphs of squares of
	cubes), we need to appropriately define the shifting operation and the notion of VC-dimension, that satisfy
	the
	conditions (1)-(3). Additionally, the degeneracy of the 1,2-inclusion graph of the final shifted family must be
	bounded by a function of the VC-dimension. We will use the shifting operation with respect to pairs of
	elements (and not to
	single elements) and the notion of clique-VC-dimension instead of VC-dimension.
	
	\subsection{Other results}
	The inequality of Haussler et al. \cite{HaLiWa} as well as the notion of VC-dimension and Sauer lemma have
	been subsequently extended to subgraphs of Hamming graphs (i.e., from binary alphabets to arbitrary
	alphabets); see \cite{HaLo,Pollard_book,Natarajan,RuBaRu}. Cesa-Bianchi and Haussler \cite{CBHa} presented
	a graph-theoretical generalization of the Sauer lemma for the $m$-fold $F^m=F\times \cdots\times F$
	Cartesian products of arbitrary undirected graphs $F$. In \cite{ChLaRa_products}, we defined a notion of
	VC-dimension for subgraphs of Cartesian products of arbitrary connected graphs (hypercubes are Cartesian
	products of $K_2$) and we established a density result $\frac{|E|}{|V|}\le \vcd(G)\cdot \alpha(H)$ for 	
	subgraphs $G$ of Cartesian products of graphs not containing a fixed graph $H$ as a minor ($\alpha(H)$ is a
	constant such that any graph not containing $H$ as a minor has density at most $\alpha(H)$; it is well known
	\cite{Die} that if $r:=|V(H)|$, then $\alpha(H)\le cr\sqrt{\log r}$ for a universal constant $c$).
	
	For edge- and vertex-isoperimetric problems in Johnson graphs (which are still open problems), some authors
	\cite{AlCa, DiSeVe}
	used a natural {\it pushing
		to the left} (or \emph{switching}, or \emph{shifting}) operation. Let $\mcs$ consists only of sets of size $r$.
	Given an arbitrary total order	$e_1,\ldots, e_m$ of the elements of $X$ and two elements $e_i<e_j$, in the
	pushing to
	the left of $\mcs$ with respect to the pair $e_i,e_j$ each set $S$ of $\mcs$ containing $e_j$ and not
	containing $e_i$ is replaced by the set $S\setminus \{ e_j\} \cup \{ e_i\}$ if $S\setminus \{ e_j\} \cup \{ e_i\}\notin
	\mcs$. This operation preserves the size of $\mcs$, the cardinality $r$ of the sets and do not decrease the
	number
	of edges, but the degeneracy of the final graph is not easy to bound.
	
	Bousquet and Thomass\'e \cite{BouTh} defined the notions of 2-shattering and 2VC-dimension and 	
	established the Erd\"os-P\'osa property for the families of balls of fixed radius in graphs with bounded
	2VC-dimension. These notions have some similarity with our concepts of c-shattering and
	clique-VC-dimension because they concern shattering not of all subsets but only of a certain pattern of
	subsets (of all pairs).
	Recall from \cite{BouTh} that a set family $\mcs$ {\it 2-shatters} a set $Y$ if for any $2$-set $\{ e_i,e_j\}$ of
	$Y$ there exists $S\in \mcs$ such that $Y\cap S=\{ e_i,e_j\}$; the {\it 2VC-dimension} of $\mcs$ is the
	maximum size of a 2-shattered set.
	
	Halved cubes and Johnson graphs host several important classes of graphs occurring from metric graph
	theory \cite{BaCh_survey}: basis graphs of matroids are isometric subgraphs of Johnson graphs \cite{Mau}
	and basis graphs of even $\Delta$-matroids are isometric subgraphs of halved cubes \cite{Ch_delta}. More
	general classes are the graphs isometrically embeddable into halved cubes and Johnson graphs. Similarly
	to Djokovi\'{c}'s characterization of isometric subgraphs of hypercubes \cite{Dj}, isometric subgraphs of
	Johnson graphs have been characterized in \cite{Ch_Johnson} (the problem of characterizing isometric
	subgraphs of halved cubes has been raised in \cite{DeLa} and is still open).
	Shpectorov \cite{Shp} proved that the graphs admitting an isometric embedding into an $\ell_1$-space are
	exactly the graphs which admit a scale embedding into a hypercube and he proved that such graphs are
	exactly the graphs which are isometric subgraphs of Cartesian products of octahedra and of isometric
	subgraphs of halved cubes. For a presentation of most of these results, see the book by Deza and Laurent
	\cite{DeLa}.
	
	\section{Preliminaries} \label{preliminaries}
	
	\subsection{Degeneracy}
	All graphs $G=(V,E)$ occurring in this note are finite, undirected, and simple. The \emph{degeneracy} of $G$
	is the minimal $k$ such that there exists a total order $v_1,\ldots,v_n$ of vertices of $G$ such that each 	
	vertex $v_i$ has degree at most $k$ in the subgraph of $G$ induced by $v_i, v_{i+1},\ldots, v_n$.
	It is well known and it can be easily shown that the degeneracy of every graph $G=(V,E)$ upper bounds the
	ratio $\frac{|E|}{|V|}$. 
	
	\subsection{Squares of hypercubes, halved cubes, and Johnson graphs}
	The \emph{$m$-dimensional hypercube} $Q_m$ is the graph having all $2^m$ subsets of a set $X=\{ 	
	e_1,\ldots,e_m\}$ as the vertex-set and two sets $A,B$ are adjacent in $Q_m$ iff $|A\triangle B|=1$. The {\it
		halved
		cube} $\frac{1}{2}Q_m$ \cite{BrCoNe,DeLa} has the subsets of $X$ of even cardinality as vertices and two
	such vertices $A,B$ are adjacent in $\frac{1}{2}Q_m$ iff $|A\triangle B|=2$ (one can also define halved cubes
	for
	subsets of odd size). Equivalently, the halved cube $\frac{1}{2}Q_{m}$ is the {\it square} $Q^2_{m-1}$ of the
	hypercube $Q_{m-1}$, i.e., the graph formed by connecting pairs of vertices of $Q_{m-1}$ whose distance is
	at most two in $Q_{m-1}$. For an integer $r>0$, the {\it Johnson graph} $J(r,m)$ \cite{BrCoNe,DeLa} has the
	subsets of $X$ of size $r$ as vertices and two such vertices $A,B$ are adjacent in $J(r,m)$ iff $|A\triangle
	B|=2$. All Johnson graphs $J(r,m)$ are (isometric) subgraphs of the corresponding halved cube
	$\frac{1}{2}Q_m$. Notice also that the halved cube $\frac{1}{2}Q_m$ and the Johnson graph $J(r,m)$ are
	scale 2 embedded in the hypercube $Q_m$.
	
	Let $\mcs$ be a family of subsets of a set $X=\{ e_1,\ldots,e_m\}$. The {\it 1,2-inclusion graph}
	$\incgraph(\mcs)$ of $\mcs$ is the graph having $\mcs$ as the vertex-set and in which two vertices $A$ and
	$B$ are adjacent
	iff $1\le |A\Delta B|\le 2$, i.e., $\incgraph(\mcs)$ is the subgraph of the square $Q_m^2$ of $Q_m$ induced by
	$\mcs$. The graph $\incgraph(\mcs)$ comprises all edges of the 1-inclusion graph $G_1(\mcs)$ of $\mcs$ 
	and of the subgraphs of the halved cubes induced by even and odd sets of $\mcs$.
	The latter edges of $\incgraph(\mcs)$ are of two types: {\it vertical edges} $SS'$ arise from sets $S,S'$ such
	that $|S|=|S'|+2$ or $|S'|=|S|+2$ and {\it horizontal edges} $SS'$ arise from sets $S,S'$ such that $|S|=|S'|$.
	
	If all sets of $\mcs$ have even cardinality, then we will call $\mcs$  an {\it even set family}; in this case,
	the 1,2-inclusion graph $\incgraph(\mcs)$ coincides with the
	subgraph of the halved cube $\frac{1}{2}Q_m$ induced by $\mcs$.
	Since $Q^2_m$ is isomorphic to
	$\frac{1}{2}Q_{m+1}$, any 1,2-inclusion graph is an induced subgraph of a halved cube. More precisely,
	any set family $\mcs$ of  $X$ can be lifted to an even set family $\mcs^+$ of
	$X\cup \{ e_{m+1}\}$ in  such a way that the 1,2-inclusion graphs of $\mcs$ and $\mcs^{+}$ are
	isomorphic: $\mcs^+$ consists of all sets of even size of $\mcs$ and of all sets of odd size of $\mcs$ to
	which the element
	$e_{m+1}$ was added. The proof of the following lemma is straightforward:
	
	\begin{lemma} \label{odd-even} For any set family $\mcs$, the lifted family $\mcs^+$  is an even set family
		and the
		1,2-inclusion graphs $\incgraph(\mcs)$ and $\incgraph(\mcs^+)$ are isomorphic.
	\end{lemma}

	\subsection{Pointed set families and pointed cliques}
	
	We will call a set family $\mcs$ a {\it pointed set family} if $\varnothing\in \mcs$. Any set family $\mcs$ can be
	transformed into a pointed set family by the operation of twisting. For a set $A\in \mcs$, let $\mcs\triangle
	A:=\{ S\triangle A: S\in \mcs\}$ and say that $\mcs\triangle A$ is obtained from $\mcs$ by applying a {\it
		twisting} with respect to $A$. Note that a twisting is a bijection between $\mcs$ and $\mcs\triangle A$
	mapping the set $A$ to $\varnothing$ (and therefore $\mcs\triangle A$ is a pointed set family). Notice that
	any twisting of an even set family $\mcs$ is an even set family. As before, let
	$G_1(\mcs)$ denote the 1-inclusion graph of $\mcs$. The following properties of twisting are well-known and
	easy to prove:
	
	\begin{lemma} \label{twisting0}
		For any $\mcs\subseteq 2^X$ and any $A\subseteq X$, $G_1(\mcs\triangle A)\backsimeq G_1 (\mcs)$ 	
		and $\vcd(\mcs\triangle A)=\vcd(\mcs)$.
	\end{lemma}
	
	Analogously to the proof of the first assertion of Lemma \ref{twisting0}, one can easily show that:
	
	\begin{lemma} \label{twisting}
		For any set family $\mcs\subseteq 2^X$ and any $A\subseteq X$, $\incgraph(\mcs\triangle A)\backsimeq
		\incgraph(\mcs)$.
	\end{lemma}
	
	We will say that a clique $\mathcal C$ of $\frac{1}{2}Q_m$  is a {\it pointed clique} if
	$\mathcal C$ is a pointed set family.
	
	\begin{lemma} \label{canonical_clique}
		By a twisting, any clique $\mathcal C$ of $\frac{1}{2}Q_m$ can be transformed into a pointed clique.
	\end{lemma}
	
	\begin{proof}
		Let $\mathcal C$ be a clique of $\frac{1}{2}Q_m$. Let $A$ be a set of maximal size which is a vertex of
		$\mathcal
		C$. Then the twisting of ${\mathcal C}$ with respect to $A$ maps $\mathcal C$ into a pointed clique
		${\mathcal C}\triangle A$ of $\frac{1}{2}Q_m$: indeed, if $C',C''$ are two vertices of $\mathcal C$, then 	
		$|(C'\triangle A)\triangle(C''\triangle A)|=|C'\triangle C''|=2$.
		Since $A\Delta A=\varnothing$, ${\mathcal C}\triangle A$ is a pointed clique (for an illustration, see Fig.
		\ref{fig_twist}).
	\end{proof}
	
	\begin{figure}
		\centering
		\includegraphics[width=0.4\linewidth]{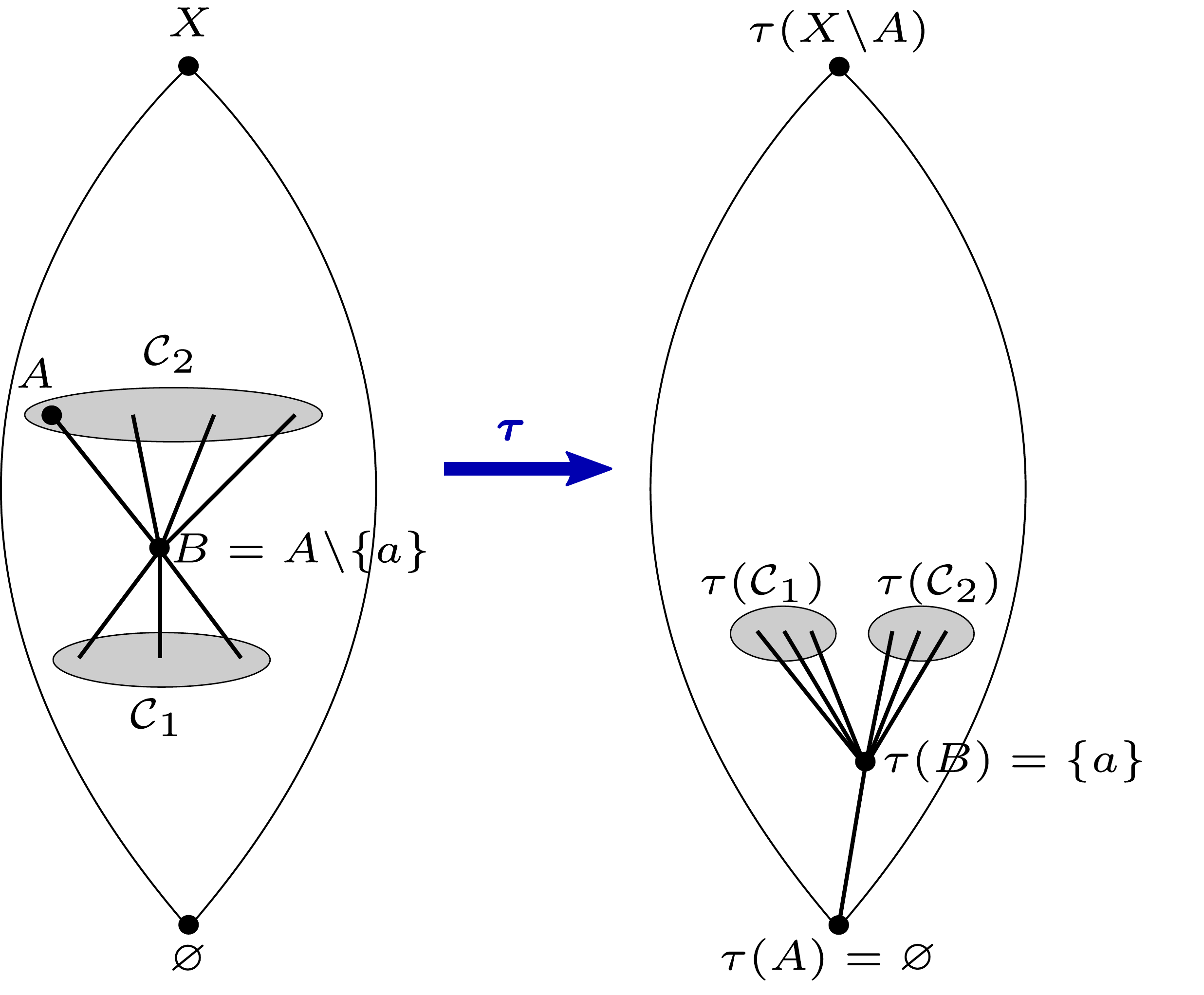}
		\caption{\label{fig_twist} A twisting mapping $\tau:S\mapsto S\triangle A$ of a clique to a pointed clique.}
	\end{figure}	
	
	We describe now the structure of pointed cliques in halved cubes.
	
	\begin{lemma} \label{lem_clique_halved-cube}
		Any pointed maximal clique $\mathcal C$ of a halved cube $\frac{1}{2}Q_m$ is (a) a sporadic 4-clique of the
		form $\{ \varnothing, \{ e_i,e_j\}, \{ e_i,e_k\},\{ e_j,e_k\}\}$ for arbitrary elements $e_i,e_j,e_k\in X$, or (b) a
		clique of size $m$ of the form $\{ \varnothing\} \cup \{ \{ e_i,e_j\}: e_j\in X\setminus \{ e_i\}\}$ for an arbitrary
		but fixed element $e_i\in X$.
	\end{lemma}
	
	\begin{proof}
		Since $\mathcal C$ is a pointed clique, $\varnothing$ is a vertex of $\mathcal C$, denote it $C_0$. All other
		neighbors of $C_0$ in $\frac{1}{2}Q_m$ are sets of the form $\{e_i,e_j\}$ with $e_i,e_j\in X$, i.e., the
		neighborhood of $C_0$ in the halved cube $\frac{1}{2}Q_m$ is the line-graph of the complete graph $K_m$
		having $X$ as the vertex-set. In particular, the clique	 ${\mathcal C}_0:={\mathcal C}\setminus \{ C_0\}$
		corresponds to a set of pairwise incident edges of $K_m$. It can be easily seen that this set of edges
		defines either a triangle or a star of $K_m$. Indeed, pick an edge $e_ie_j$ of $K_m$ corresponding to a pair
		$\{ e_i,e_j\}\in {\mathcal C}_0$. If the respective set of edges is not a star, then necessarily ${\mathcal C}_0$
		contains two pairs of the form $\{e_i,e_k\}$ and $\{ e_j,e_l\}$, both different from $\{ e_i,e_j\}$. But then
		$k=l$, otherwise the edges $e_ie_k$ and $e_je_l$ would not be incident. Thus ${\mathcal C}_0$ contains
		the three pairs $\{ e_i,e_j\},\{ e_i,e_k\}$, and $\{ e_j,e_k\}$. If ${\mathcal C}_0$ contains yet another pair, then
		this pair will be necessarily disjoint from one of the three previous pairs, a contradiction. Thus in this case,
		${\mathcal C}=\{ \varnothing, \{ e_i,e_j\}, \{ e_i,e_k\},\{ e_j,e_k\}\}$. Otherwise, if the respective set of edges is
		a star with center $e_i$, then ${\mathcal C}_0$ is a clique of size $m-1$ of the form $\{ \{ e_i,e_j\}: e_j\in
		X\setminus \{ e_i\}\}$.
	\end{proof}

	\section{The clique-VC-dimension}
	As we noticed above, the classical VC-dimension of set families cannot be used to bound the density of their
	1,2-inclusion graphs. Indeed, the 1,2-inclusion graph of the set family $\mcs_0:=\{ \{ e_j\}: e_j\in X\}$ is
	a complete graph, while the VC-dimension of $\mcs_0$ is $1$ (notice also that the 2VC-dimension of
	$\mcs_0$ is $0$).

	We will define a notion that is more appropriate for this purpose, which we will call clique-VC-dimension.
	The idea is to use the form of pointed cliques of $\frac{1}{2}Q_m$ established above and to shatter them. In
	view of
	Lemma \ref{odd-even}, it suffices to define the clique-VC-dimension for even set families.
	First we present a generalized
	 definition of classical shattering.
	
	Let $X=\{ e_1,\ldots,e_m\}$ and $\mcs\subseteq 2^X$. Let $Y$ be a subset of $X$. Denote by $Q[Y]$ the
	subcube
	of $Q_m$ consisting of all subsets of $Y$. Analogously, for two sets $Y'$ and $Y$ such that $Y'\subset Y$,
	denote by $Q[Y',Y]$ the smallest subcube of $Q_m$ containing the sets $Y'$ and $Y$: $Q[Y',Y]=\{ Z\subset
	X: Y'\subseteq Z\subseteq Y\}$. In particular, $Q[Y]=Q[\varnothing,Y]$. For a vertex $Z$ of $Q[Y',Y]$, call
	$$
	F(Z):=\{ Z\cup Z': Z'\subseteq X\setminus Y\}
	$$
	the {\it fiber} of $Z$ with respect to the cube $Q[Y',Y]$. Let
	$$
	\pi_{Q[Y',Y]}(\mcs):=\{ Z\in Q[Y',Y]: F(Z)\cap \mcs\ne \varnothing\}
	$$
	denote the {\it projection} of the set family $\mcs$ on $Q[Y',Y]$. Then the cube $Q[Y',Y]$ with $Y'\subseteq
	Y$ is shattered by $\mcs$ if $\pi_{Q[Y',Y]}(\mcs)=Q[Y',Y]$, i.e., for any $Y'\subseteq Z\subseteq Y$ the fiber
	$F(Z)$ contains a set of $\mcs$ (see Fig. \ref{fig_shattering}). In particular, a subset $Y$ is shattered by
	$\mcs$ iff $\pi_{Q[Y]}(\mcs)=Q[Y]$.
	
	\subsection{The clique-VC-dimension of pointed even set families} Let $\mcs$ be a pointed even set family
	of $2^X$, i.e., a set family in which all sets have even size and $\varnothing\in \mcs$.
	Let $Y$ be a subset of $X$ and let $e_i$ be an element of $X$ not belonging to $Y$. Denote by $P(e_i,Y)$ the
	set of all {\it $2$-sets}, i.e., pairs of the form $\{e_i,e_j\}$ with $e_j\in Y$. Then $Q[\{ e_i\},Y\cup\{
	e_i\}]$ is the smallest subcube of $Q_m$ containing $e_i$ and all the $2$-sets of $P(e_i,Y)$. For simplicity,
	we will denote this cube by $Q(e_i,Y)$.
	
	We will say that a pair $(e_i,Y)$ with $Y\subset X$ and $e_i\notin
	Y$ is {\it c-shattered} by $\mcs$ if there exists a surjective function $f: \pi_{Q(e_i,Y)}(\mcs)\rightarrow P(e_i,Y)$
	such that for any $S\in \pi_{Q(e_i,Y)}(\mcs)$ the inclusion $f(S)\subseteq S$ holds. In other words, $(e_i,Y)$ is
	c-shattered by $\mcs$ if each $2$-set $\{e_i,e_j\}\in P(e_i,Y)$ admits an extension $S_j\in \pi_{Q(e_i,Y)}(\mcs)$
	such that $\{ e_i,e_j\}\subseteq S_j$ and for any two 2-sets $\{e_i,e_j\}, \{e_{i},e_{j'}\}\in P(e_i,Y)$ the sets $S_j$
	and $S_{j'}$ are distinct. Since $\varnothing\in \mcs$, the empty set $\varnothing$ is always shattered by
	$\mcs$.
	
	\begin{figure}
		\centering
		\includegraphics[width=0.5\linewidth]{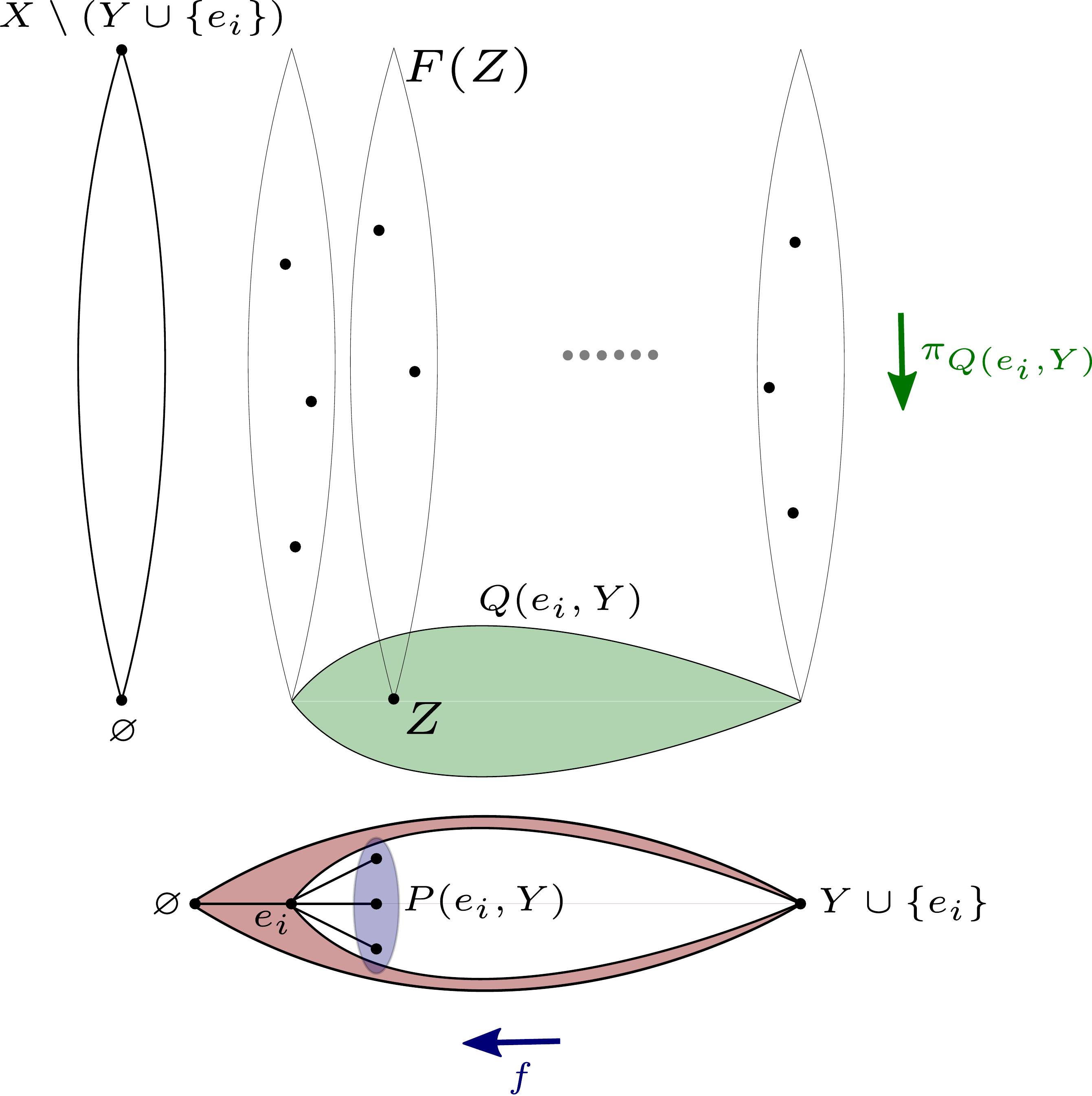}
		\caption{\label{fig_shattering}
			Example of a c-shattered pair $(e_i,Y)$. $F(Z)$ is the fiber of $Z$ in $Q(e_i,Y)$. The sets of $\mcs$
			(black points) in the fibers of the sets of $Q(e_i,Y)$ are projected on $Q(e_i,Y)$ (in green). The vertices
			in $Q(e_i,Y)$ are then mapped to $P(e_i,Y)$ (in blue) by the c-shattering function $f$. The remaining vertices
			of $Q[\varnothing,Y\cup\{e_i\}]$ (in red) are ``not used'' for shattering.
		}
	\end{figure}	
	
	For a pointed even set family $\mcs$, the {\it clique-VC-dimension} is
	$$
	\vccdim(\mcs):=\max\{ |Y|+1: Y\subset X \mbox{ and } \exists e_i\in X\setminus Y \mbox{ such that }
	(e_i,Y) 	
	\mbox{ is c-shattered by } \mcs\}.
	$$

	We continue with some simple examples  of clique-VC-dimension:
	
	\begin{example} \label{example1} For set family $\mcs_0=\{ \{ e_j\}: e_j\in X\}$ introduced above, let
    $\mcs^+_0=\{ \{ e_j,e_{m+1}\}: e_j\in X\}$ be the lifting of $\mcs_0$ to an even set family. For an
	arbitrary (but fixed) element $e_i$, let $\mcs_1:=\{ \varnothing\}\cup \{ \{ e_i,e_j\}: e_j\ne e_i\}$.
    Then $\mcs_1$ coincides with $\mcs_0\Delta \{ e_i\}$ and with $\mcs^+_0\Delta \{ e_i,e_{m+1}\}$.
    $\mcs_1$ is an even set family, its 1,2-inclusion graph is a pointed clique, and $\vccdim(\mcs_1)=|X|=m$.
	\end{example}
	
	\begin{example} \label{example3} Let $\mcs_2=\{ \varnothing, \{ e_1,e_2\}, \{ e_1,e_3\},\{ e_2,e_3\}\}$ be the
		sporadic 4-clique from
		Lemma \ref{lem_clique_halved-cube}. In this case, one can c-shatter any two of the pairs $\{ e_1,e_2\}, \{
		e_1,e_3\},\{ e_2,e_3\}$ but
		not all three. This shows that $\vccdim(\mcs_2)=2+1=3$.
	\end{example}
	
	\begin{example} \label{example4} For arbitrary even integers $m$ and $k$, Let $X$ be a ground set of size
		$m+km$ which is the disjoint union of $m+1$ sets $X_0,X_1,\ldots,X_m$, where $X_0=\{ e_1,\ldots,e_m\}$
		and
		$X_i=\{ e_{i1},\ldots,e_{ik}\}$ for each $i=1,\ldots,m$.
		Let $\mcs_3$ be the pointed even set family consisting of the empty set $\varnothing$, the set $X$, and
		for
		each $i=1,\ldots,m$ of all the $2$-sets of $P(e_i,X_i)=\{ \{ e_i,e_{i1}\},\ldots, \{ e_i,e_{ik}\}\}$.
		Then $\incgraph(\mcs_3)$ consists of an isolated vertex $X$ and $m$ maximal cliques ${\mathcal
			C}_i:=P(e_i,X_i)\cup \{ \varnothing\}$ of size $k+1$ and these cliques pairwise intersect in a single vertex
		$\varnothing$. We assert that $\vccdim(\mcs_3) = k+2$. Indeed, let $Y$ be the set consisting of $X_i$ for a
		given $i\in\{1,\ldots,m\}$ plus the singleton $\{e_{(i+1)1}\}$. Then the pair $(e_i,Y)$ is c-shattered by
		$\mcs_3$.
		The c-shattering map $f:\pi_{Q(e_i,Y)}(\mcs_3) \to P(e_i,Y)$ is defined as follows: every $2$-set of
		$P(e_i,X_i)\subset Q(e_i,Y)$ is in $\mcs_3$ and is thus mapped to itself, $X\cap(Y\cup\{e_i\}) = Y\cup\{e_i\}$
		is
		an extension of the remaining $2$-set $\{e_i,e_{(i+1)1}\}$ in $Q(e_i,Y)$ and thus $f(Y\cup\{e_i\}) :=
		\{e_i,e_{(i+1)1}\}$. Since $|Y|=k+1$, we showed that $\vccdim(\mcs_3)\geq k+2$.		
		On the other hand, $\vccdim(\mcs_3) \leq k+2$ because every element $e$ from $X$ is in at most $k+1$ sets of $\mcs_3$.
		Therefore, $\vccdim(\mcs_3) = k+2$.
	\end{example}

	\begin{figure}
		\centering
		\includegraphics[width=0.5\linewidth]{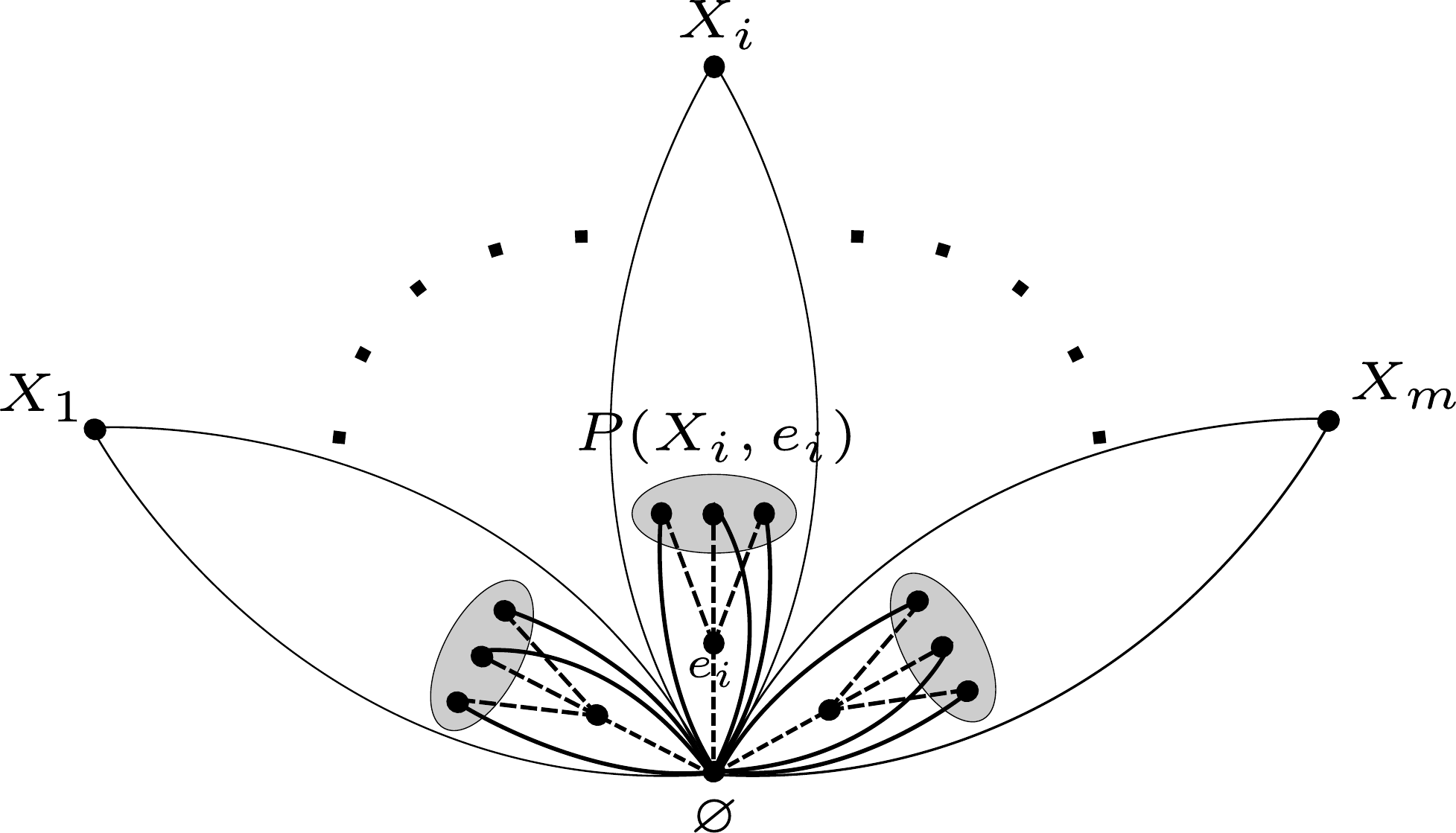}
		\caption{\label{fig_cvc-dim} Illustration of Example \ref{example4}}
	\end{figure}	
	
	\subsection{The clique-VC-dimension of even and arbitrary set families}
	The {\it clique-VC-dimension} $\vccdim^*(\mcs)$ of an even set family $\mcs$ is the minimum of the
	clique-VC-dimensions of the pointed even set families $\mcs\triangle A$ for $A\in \mcs$:
	$$
	\vccdim^*(\mcs):=\min \{ \vccdim(\mcs\triangle A): A\in \mcs\}.
	$$
	The {\it clique-VC-dimension} $\vccdim^*(\mcs)$ of an arbitrary set family $\mcs$ is the clique-VC-dimension
	of its  lifting $\mcs^+$.
	
	\begin{remark} \label{remark5} A simple analysis shows that for the even set families from Examples
		\ref{example1}-\ref{example4},
		we have $\vccdim^*(\mcs_1)=m$, $\vccdim^*(\mcs_2)=3$, and $\vccdim^*(\mcs_3)=k+2$.
	\end{remark}

	\begin{remark} In fact, the set family $\mcs_3$ shows that the maximum degree of a 1,2-inclusion graph
		$\incgraph(\mcs)$ of an even set
		family  $\mcs$ can be arbitrarily larger than $\vccdim^*(\mcs)$. Indeed, $\varnothing$ is the vertex of
		maximum degree of
		$\incgraph(\mcs_3)$ and its degree is $km$.
		
		The family $\mcs_3$ also explains why in the definition of the clique-VC-dimension of $\mcs$ we take the
		minimum
		over all $\mcs\triangle A, A\in \mcs$.  Consider
		the twisting of $\mcs_3$ with respect to the set $X\in \mcs_3$. Then one can see that
		$\vccdim(\mcs_3\triangle X)\ge (m-1)k+1$.	Indeed, $\mcs_3\Delta X=\{\varnothing,
		X\}\cup(\bigcup_{(i,j)\in\{1,\ldots,m\}\times\{1,\ldots,k\}}\{X\backslash\{e_i,e_{ij}\}\})$. Let $Y:=\{e_{ij} :
		i\in\{1,\ldots,m-1\}\text{ and } j\in\{1,\ldots,k\}\}$. We assert that $(e_1,Y)$ is c-shattered by $\mcs_3\Delta
		X$. We set $\mcs_3':=\pi_{Q(e_1,Y)}(\mcs_3\Delta X)$, $S_{ij}:=X\setminus\{e_i,e_{ij}\}$, and
		$S_{ij}':=\pi_{Q(e_1,Y)}(S_{ij})$. Let $f: \mcs_3' \to P(e_1,Y)$ be such that for all $i\in\{2,\ldots,m\}$ and
		$j\in\{1,\ldots,k\}$, we have $f(S_{ij}')=\{e_1,e_{(i-1)j}\}$.
		Clearly, every $\{e_1,e_{(i-1)j}\}$ has an extension $S_{ij}'$ with a non-empty fiber ($S_{ij}\in F(S_{ij}')$), and
		for all $S_{rl}\neq S_{ij}$, we have $S_{rl}'\neq S_{ij}'$, hence $f$ is a surjection. Therefore, $(e_1,Y)$ is
		c-shattered. Since $|Y|=(m-1)k$, whence $\vccdim(\mcs_3\triangle X)\ge (m-1)k+1$.
	\end{remark}

	\section{Proof of Theorem \ref{density}}
	After the preparatory work done in previous three subsections, here we present the proof of our main result.
	We start the proof by defining the double shifting (d-shifting) as an adaptation of the shifting to pointed even families.
	We show that, similarly to classical shifting operation, d-shifting
	satisfies the conditions (1)-(3) and that the result of a complete sequence of d-shiftings is a bouquet of
	halved cubes (which is a particular pointed even set family). We show that the degeneracy of the
	1,2-inclusion graph of such a bouquet $\bouquet$ is bounded by $\const$,
	where $d:=\vccdim(\bouquet)$. We conclude the proof of the theorem by considering arbitrary even set
	families
	$\mcs$ and applying the previous arguments to the pointed family $\mcs\Delta A$, where $A$ is a set of
	$\mcs$ such that $\vccdim(\mcs\Delta A)=\vccdim^*(\mcs)$.
	
	\subsection{Double shiftings of pointed even families}
	For a pointed even set family $\mcs\subseteq 2^X,$ the \emph{double shifting} ({\it d-shifting} for short)
	with respect to a 2-set $\{e_i,e_j\}\subseteq X$ is a map $\phiij: \mcs\rightarrow 2^X$ which replaces every set
	$S$
	of $\mcs$ such that $\{e_i,e_j\}\subseteq S$ and $S\setminus\{e_i,e_j\} \notin \mcs$ by the set
	$S\setminus\{e_i,e_j\}$:
	$$
	\begin{aligned}
	\phiij:~	& \mcs &\to & ~2^X	\\
	&S		& \mapsto	&\begin{cases}
	S\setminus\{e_i,e_j\},	&\text{if $\{e_i,e_j\}\subseteq S$ and $S\setminus\{e_i,e_j\}
		\notin \mcs$}	\\
	S,							&\text{otherwise.}
	\end{cases}.
	\end{aligned}
	$$
	
	\begin{proposition}\label{lem_push-down_augmente_densite}
		Let $\mcs \subseteq 2^X$ be a pointed even set family, let $\{e_i,e_j\} \subseteq X$ be a $2$-set, and let 	 	
		$\incgraph(\mcs) = G = (V,E)$ and $\incgraph(\phiij(\mcs)) = G' = (V',E')$ be the subgraphs of the halved
		cube
		induced by $\mcs$ and $\phiij(\mcs)$, respectively. Then $|V|=|V'|$ and $|E|\le |E'|$ hold.
	\end{proposition}
	\begin{proof}	
		The fact that a d-shifting $\phiij$ preserves the number of vertices of an induced subgraph of halved cube
		immediately follows from the definition. Therefore we only need to show that $\phiij$ cannot decrease
		the number of edges, i.e., that there exists an injective map $\psiij: E\rightarrow E'$. We will call an edge
		$SS'$ of $G$ \emph{stable}	if $\phiij(S)=S$ and $\phiij(S')=S'$ hold and \emph{shiftable}
		otherwise. For each stable edge $SS'$ we will set $\psiij(SS'):=SS'$.
		
		Now, pick any shiftable edge $SS'$ of $E$. Notice that in this case $\{e_i,e_j\}\subseteq S$ or
		$\{e_i,e_j\}\subseteq
		S'$. To
		define $\psiij(SS')$, we distinguish two cases depending on whether $\{e_i,e_j\}$ is a subset of only one of
		the sets
		$S$, $S'$ or of both of them.
		
		\medskip\noindent
		{\bf Case 1$'$.} $\{e_i,e_j\}\subseteq S$ and $\{e_i,e_j\}\not\subseteq S'$ (the case $\{e_i,e_j\}\subseteq S'$
		and
		$\{e_i,e_j\}\not\subseteq S$ is similar).
		
		\medskip\noindent
		Since $\{e_i,e_j\}\not\subseteq S'$, necessarily $\phiij(S')=S'$. Since $SS'$ is shiftable, $\phiij(S)\ne S$, i.e.,
		$\phiij(S)=S\setminus \{ e_i,e_j\}=:Z$. We consider two cases depending on whether one of the elements
		$e_i$ or $e_j$ belongs to $S'$ or not.
		
		\medskip\noindent {\bf Subcase 1$'$.1.} $e_i\in S'$ and $e_j\not\in S'$ (the case $e_j\in S'$ and $e_i\not\in
		S'$ is
		similar). In this case, there is an element $e_k\in X$ such that $S\Delta S'=\{e_j,e_k\}$. Observe that
		$S\not\subseteq
		S'$
		since $e_j\not\in S'$ and $e_j\in S$. Hence
		either $S'\subseteq S$ or there exists $A\subset X$ such that $S'=A\cup\{e_k\}$ and $S=A\cup\{e_j\}$. In
		the
		former case, we have $S=S'\cup\{e_j,e_k\}$, $Z=S'\cup\{e_k\}\setminus\{e_i\}$, and $Z\Delta S'=\{e_i,e_k\}$.
		In the later
		case, we have $Z=A\setminus\{e_i\}$ and
		$Z\Delta S'=\{e_i,e_k\}$. In both cases, $|Z\Delta S'|=2$ and $ZS'\in E'$. We set $\psiij(SS'):=ZS'$.
		
		\medskip\noindent {\bf Subcase 1$'$.2.} $e_i\not\in S'$ and $e_j\not\in S'$. Then $S\Delta S'=\{e_i,e_j\}$ and
		so
		$S\setminus \{ e_i,e_j\}=Z=S'$.
		We obtain a contradiction that  $SS'$ is shiftable (i.e., $Z=\phiij(S)$ cannot be
		in $\mcs$).
		
		\medskip\noindent
		{\bf Case 2$'$.} $\{e_i,e_j\}\subseteq S$ and $\{e_i,e_j\}\subseteq S'$.
		
		\medskip\noindent
		Set $Z:=S\setminus \{ e_i,e_j\}$ and $Z':=S'\setminus \{ e_i,e_j\}$. Then both sets $Z,Z'$ belong to
		$\phiij(\mcs)$
		and $ZZ'$ defines an edge of $G'$.
		Since $SS'$ is shiftable, at least one of the sets $Z,Z'$ does not belong to $\mcs$.
		
		\medskip\noindent
		{\bf Subcase 2$'$.1.} $Z,Z'\notin \mcs$. Then $\phiij(S)=Z$ and $\phiij(S')=Z'$ and $ZZ'$ is an
		edge of $G'$. In this case, we set $\psiij(SS'):=ZZ'$.
		
		\medskip\noindent
		{\bf Subcase 2$'$.2.} $Z\in \mcs$ and $Z'\notin \mcs$ (the case $Z\notin \mcs$ and $Z'\in \mcs$ is
		similar).
		Then $\phiij(S)=S,\phiij(S')=Z'$, and $ZZ'$ is an edge of $G'$
		but not  of $G$. In this case, we set $\psiij(SS'):=ZZ'$.
		
		\medskip
		It remains to show that the map $\psiij: E\rightarrow E'$ is injective. Suppose by way of contradiction that
		$G'$ contains an edge $ZZ'$ for which there exist two distinct edges $SS'$ and
		$CC'$ of $G$ such that $\psiij(SS')=\psiij(CC')=ZZ'$. Since at least one of the edges $SS'$
		and $CC'$ is different from $ZZ'$, from the definition of d-shifting we conclude that $ZZ'$ is not an edge
		of $G$, say $Z'\notin \mcs$. This also implies that $SS'$ and $CC'$ are shiftable edges of
		$G$.
		
		\medskip\noindent
		{\bf Case 1$''$.} $Z\notin \mcs$.
		
		\medskip\noindent
		From the definition of the map $\psiij$ and since $Z,Z'\notin \mcs$, both edges $SS'$ and $CC'$ are in
		Subcase 2$'$.1. This shows that $Z=S\setminus \{ e_i,e_j\}, Z'=S'\setminus \{ e_i,e_j\},$ and $Z=C\setminus \{
		e_i,e_j\},Z'=C'\setminus \{ e_i,e_j\}$, yielding $S=C$ and $S'=C'$, a contradiction.
		
		\medskip\noindent
		{\bf Case 2$''$.} $Z\in \mcs$.
		
		\medskip\noindent
		After an appropriate renaming of the sets $S,S'$ and $C,C'$, we can suppose that $\phiij(S)=\phiij(C)=Z$
		and
		$\phiij(S')=\phiij(C')=Z'$. Since $Z'\notin \mcs$, from the definition of the map $\psiij$, we deduce that
		$S'=Z'\cup \{e_i,e_j\}=C'$. On the other hand, since $Z\in \mcs$, we have either $S=C=Z$ which contradicts the choice of $SS' \neq
		CC'$, or $S\setminus\{e_i,e_j\}=C=Z$ (or the symmetric possibility $C\setminus\{e_i,e_j\}=S=Z$) which contradicts the
		fact that $SS'$ (or $CC'$) is shiftable.
		
		\medskip
		This shows that the map $\psiij: E\rightarrow E'$ is injective, thus $|E|\le |E'|$.
	\end{proof}
	
	\begin{lemma}\label{lem_push-down_baisse_vccdim2}
		If $\phiij$ is a d-shifting of a pointed even family  $\mcs\subset2^X$, then $\vccdim(\phiij(\mcs))
		\leq \vccdim(\mcs)$.
	\end{lemma}
	
	\begin{proof}

		Let $(e,Y)$ be c-shattered by $\mcsij := \phiij(\mcs)$ (recall that $Y\subset X$ and $e\notin Y$). Let  $\mcs' :=
		\pi_{Q(e,Y)}(\mcs)$ and $\mcsij' := \pi_{Q(e,Y)}(\phiij(\mcs))$. By definition of c-shattering, there exists a surjective function $f$ associating every element of $\mcsij'$ to a $2$-set $\{e,e'\}\in P(e,Y)$.
		We will define a surjective function $g$ from $\mcs'$ to $\mcsij'$, and derive from $f$ a c-shattering function
		$f':=f\circ g$ from $\mcs'$ to $P(e,Y)$. Let $S_{e'}\in \mcsij'$ be a set such that $f(S_{e'})=\{e,e'\}$.
		If $S_{e'}\in\mcs'$, then the $2$-set $\{e,e'\}$ also has an extension in $\mcs'$ and we can set $g(S_{e'}) := S_{e'}$.
		If $S_{e'}\not\in\mcs'$, it means that there exists a set $S\in\mcs$ such that $S\neq\phiij(S)$ and $\phiij(S)$ is in the fiber
		$F(S_{e'})$ of $S_{e'}$ with respect to $\pi_{Q(e,Y)}$ in $\mcsij$. The set $S$ is in the fiber $F(S')$ of some set $S'\in\mcs'$
		with
		respect to $\pi_{Q(e,Y)}$. Since $\phiij(S)\subseteq S$, we have $S_{e'}\subseteq S'$ and $S'\in \mcs'$ is an extension of
		the $2$-set $\{e,e'\}$.
		We set $g(S') := S_{e'}$. Moreover, for every set $S'\in \mcs'\setminus\mcsij'$, there is a set $S\in F(S')$ such that $\phiij(S)\neq S$. In this case, there is a set $S_{e'}\in \mcsij'$ such that $\phiij(S)\in F(S_{e'})$. We set $g(S'):=S_{e'}$. We have $S_{e'}\subseteq S'$ since $\phiij(S)\subset S$.

		The function $g$ is surjective by definition and maps every set of $\mcs'$ either on itself or on a subset of it. Since $f$ is a
		c-shattering function, so is $f' := f\circ g$ and $(e,Y)$ is c-shattered by $\mcs$. Consequently,  we have
		$\vccdim(\phiij(\mcs)) \leq \vccdim(\mcs)$ since every $(e,Y)$ c-shattered by $\mcsij$ is also	c-shattered by $\mcs$.
	\end{proof}
	
	\begin{figure}
		\centering
		\includegraphics[width=0.55\textwidth]{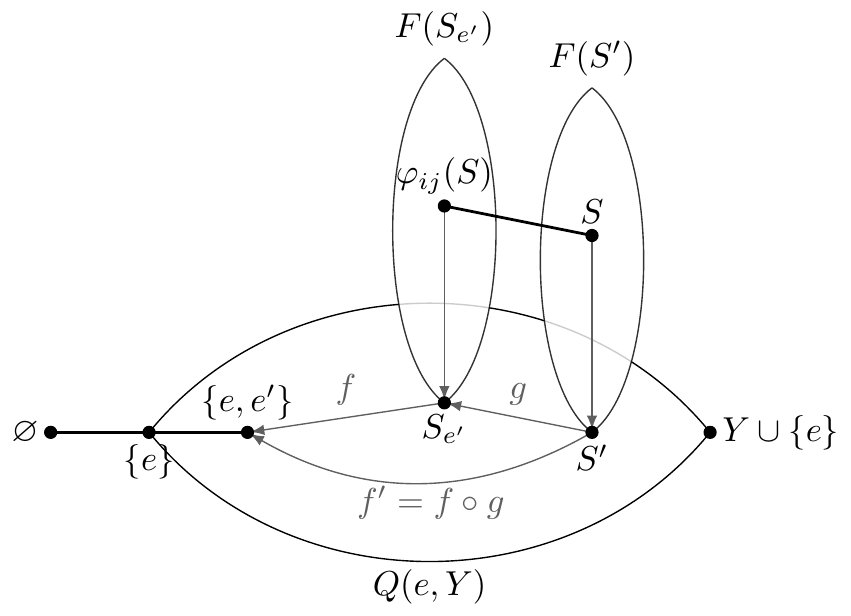}
		\caption{\label{fig_lemma_6} To the proof of Lemma~\ref{lem_push-down_baisse_vccdim2}.}
	\end{figure}

	\subsection{Bouquets of halved cubes}
	A {\it bouquet of cubes} (called usually a {\it downward closed family} or a {\it simplicial complex}) is a set
	family $\bouquet\subseteq 2^X$ such that $S\in \bouquet$ and $S'\subseteq S$ implies $S'\in \bouquet$.
	Obviously $\bouquet$ is a pointed family. Note that any bouquet of cubes $\bouquet$ is the union of all
	cubes of the form $Q[\varnothing,S]$, where $S$ is an inclusion-wise maximal subset of $\bouquet$.
	
	A {\it bouquet of halved cubes} is an even set family $\bouquet\subseteq 2^X$ such that for any $S\in
	\bouquet$, any subset $S'$ of $S$ of even size is included in  $\bouquet$. In other words, a bouquet of
	halved cubes
	$\bouquet$ is the union of all halved cubes spanned by $\varnothing$ and inclusion-wise maximal subsets $S$ of $\bouquet$.

	\begin{lemma} \label{punctured-bouquet} After a finite number of d-shiftings, any pointed even set
		family $\mcs$ of $2^X$ can be transformed into a bouquet of halved cubes.
	\end{lemma}
	
	\begin{proof} Let $\mcs_0,\mcs_1,\mcs_2,\ldots$ be a sequence of even set families such that $\mcs_0=\mcs$
		and, for any $i\ge 1$, $\mcs_{i}$ was obtained from $\mcs_{i-1}$ by a d-shifting and $\mcs_i\ne \mcs_{i-1}$.
		This sequence  is necessarily finite because each d-shifting strictly decreases the sum of sizes of
		the sets in the family.  Let $\mcs_r$ denote the last family in the sequence. This means that the d-shifting of
		$\mcs_r$ with respect to any pair of elements of $X$ leads to the same set family $\mcs_r$. Therefore, for
		any
		set $S\in \mcs_r$ and for any pair $\{ e_j,e_k\}\subseteq S$, the set $S\setminus \{ e_j,e_k\}$ belongs to
		$\mcs_r$, i.e.,
		$\mcs_r$ is a bouquet of halved cubes.
	\end{proof}

	We continue with simple properties of bouquets of halved cubes.
	
	\begin{lemma} \label{bouquet-properties}
		Let $\bouquet\subset2^X$ be a bouquet of halved cubes of \vccdimt $d:=\vccdim(\bouquet)$. Then the
		following properties hold:
		\begin{itemize}
			\item[(i)] for any element $e_i\in X$, $|\{\{e_i,e_j\}\in\bouquet : e_j\in X\setminus\{e_i\}\}| \le d-1$;
			\item[(ii)] if $S$ is a set of $\bouquet$, then $|S| \le d$;
			\item[(iii)] if $S$ is a set of $\bouquet$ maximal by inclusion, then $\bouquet\setminus\{S\}$ is still a
			bouquet of halved cubes.
		\end{itemize}
	\end{lemma}
	
	\begin{proof}
		The inequality $|\{\{e_i,e_j\}\in\mcs : e_j\in X\setminus\{e_i\}\}| \leq d-1$ directly follows from the
		definition of $\vccdim(\bouquet)$. The property (iii) immediately follows from the definition of a bouquet of
		halved cubes. To prove (ii), suppose by way of contradiction that $|S|>d$. Since $\bouquet$ is a
		bouquet of halved cubes, every subset of $S$ of even cardinality belongs to $\bouquet$. Therefore, if we
		pick any $e\in S$ and if we set $Y:=S\setminus \{e\}$, then all the $2$-sets of the form $\{e,e'\}$ with
		$e'\in	Y$ are subsets of $S$, and thus are sets of $\bouquet$. Consequently $(e,Y)$ is c-shattered by
		$\bouquet$. Since $|Y|=|S|-1>d-1$, this contradicts the assumption that $d=\vccdim(\bouquet)$.
	\end{proof}
	
	\subsection{Degeneracy of bouquets of halved cubes} 	
	In this subsection we prove the following upper bound for degeneracy of 1,2-inclusion graphs of bouquets of
	halved cubes:
	
	\begin{proposition} \label{bouquet-density}
		Let $\bouquet\subset2^X$ be a bouquet of halved cubes of \vccdimt $d:=\vccdim(\bouquet)$, and let $G :=
		\incgraph(\bouquet)$. Then the degeneracy of $G$ is at most $\binom{d}{2}$.
	\end{proposition}
	
	\begin{proof}
		Let $S$ be a set of maximal size of $\bouquet$. 
		By Lemma \ref{bouquet-properties}(iii), $\bouquet\setminus\{S\}$ is a bouquet of halved cubes. Thus, it
		suffices to show that the degree of $S$ in $G$ is upper bounded by $\binom{d}{2}$. From Lemma 		
		\ref{bouquet-properties}(ii), we know that $|S| \le d$. This implies that $S$ is incident in $G$ to at most
		$\binom{d}{2}$ vertical edges. Therefore, it remains to
		bound the number of horizontal edges sharing $S$. The following lemma will be useful for this purpose:
		
		\begin{lemma} \label{bouquet-aux}
			If $|S|=d-k\le d$, then $S$ is incident in $G$ to at most $(d-k)k$ horizontal edges.
		\end{lemma}
		
		\begin{proof}
			Pick any $s\in S$ and set $Y:=S\setminus \{s\}$. For an element $e\in X\setminus S$, let $S^e_s := Y\cup
			\{e\}$. Notice that such $S^e_s$ are exactly the neighbors of $S$ in $\frac{1}{2}Q_m$ connected by
			a horizontal edge.  Let $X' = \{e\in X\setminus S: S^e_s\in \bouquet\}$.
			
			Pick any element $y \in Y$. Then $y \in S^e_s$ for any $e \in X'$. Since $\bouquet$ is a bouquet of
			halved cubes, each of the $d-k-1$ pairs $\{y,e'\}$ with $e' \in S\setminus\{y\}$ belongs to
			$\bouquet$ (yielding $P(y,S\setminus\{y\}) \subseteq \bouquet$). To each set $S^e_s, e\in X',$
			corresponds the unique pair $\{y,e\}$ and $\{y,e\} \in \bouquet$ because $y,e \in S^e_s$. Therefore
			$P(y,X')\subset\bouquet$. Since $|P(y,S\setminus \{ y\})|+|P(y,X')|\le d-1$ and
			$|P(y,S\setminus\{y\})| = d-k-1, |X'| = |P(y,X')|$, we conclude that $|X'| \le k$. Therefore, for a fixed
			element $s \in S$, $S$ has at most $k$ neighbors of the form $S^e_s$ with $e \in X'$. Since there are $|S| =
			d-k$ possible choices of the element $s$, $S$ has at most $(d-k)k$ neighbors of cardinality $|S|$.
		\end{proof}
		
		We now continue the proof of Proposition \ref{bouquet-density}.
		Let $|S| = d-k \le d$. Then $S$ has $\binom{d-k}{2}$ neighbors of the form
		$S\backslash\{e,e'\}$ with $e\neq e'\in S$, i.e., $S$ has $\binom{d-k}{2}$ incident vertical edges. It
		remains	to bound the number of neighbors of $S$ of the form $S\backslash\{e\}\cup\{e'\}$ with $e\in S$
		and $e'\in X\setminus S$. By Lemma \ref{bouquet-aux}, $S$ has at most $(d-k)k$ such neighbors.
		Summarizing, $S$ possesses $(d-k)k + \binom{d-k}{2} = 		
		\frac{1}{2}(d^2-d-k^2+k)$ neighbors in $G$, and this number is maximal for $k=0$ because 
		$$
			\frac{1}{2}(d^2- d - k^2 + k)
			=\frac{1}{2}(d^2 - d) - \frac{1}{2}(k^2-k)
			= \binom{d}{2} - \binom{k}{2} \le \binom{d}{2}.
		$$
		Hence, the degree of $S$ in $G$ is at most $\binom{d}{2}$, as asserted.
	\end{proof}
	
	\subsection{Proof of Theorem \ref{density}}
	First, let $\mcs$ be an even set family over $X$ with $|X|=m$, $d=\vccdim^*(\mcs)$ be the
	clique-VC-dimension of $\mcs$, and $\incgraph(\mcs)=(V,E)$ be the 1,2-inclusion graph of $\mcs$. We have
	to prove that $\frac{|E|}{|V|}\le \binom{d}{2}=:D.$
	
	Let $A$ be a set of $\mcs$ 	such that $\vccdim(\mcs\Delta A)=\vccdim^*(\mcs)=d$. By Lemma \ref{twisting},
	$\incgraph(\mcs\triangle A)\backsimeq \incgraph (\mcs)$. Thus it suffices to prove the inequality
	$\frac{|E(\incgraph(\mcs\triangle A))|}{|V(\incgraph(\mcs\triangle A))|}\le D.$ Consider a complete sequence of
	d-shiftings of $\mcs\triangle A$ and denote by $(\mcs\triangle A)^*$ the resulting set family. Since
	$\mcs\triangle A$ is a pointed even set family, applying Lemma \ref{lem_push-down_baisse_vccdim2} to each
	d-shifting, we deduce that $\vccdim((\mcs\triangle A)^*)\leq\vccdim(\mcs\triangle A)=d$.
	By Lemma \ref{punctured-bouquet}, $(\mcs\triangle A)^*$ is a bouquet of halved cubes, thus, by Proposition
	\ref{bouquet-density}, the degeneracy of its 1,2-inclusion graph $\Gbouquet=\incgraph((\mcs\triangle A)^*)$
	is at most $D$. Therefore, if $\Gbouquet=(\Vbouquet,\Ebouquet)$, then
	$\frac{|\Ebouquet|}{|\Vbouquet|}\le D$ (here we used the fact that the degeneracy of a graph $G=(V,E)$ is
	an upper bound for the ratio $\frac{|E|}{|V|}$). Applying Proposition \ref{lem_push-down_augmente_densite}
	to each of the d-shiftings and taking into account that $\incgraph(\mcs\triangle A)\backsimeq \incgraph
	(\mcs)$, we conclude that $\frac{|E|}{|V|}\le \frac{|\Ebouquet|}{|\Vbouquet|}$, yielding the required density
	inequality $\frac{|E|}{|V|}\le D$ and finishing the proof of Theorem \ref{density} in case of even set families.
	If $\mcs$ is an arbitrary set family, then $\vccdim^*(\mcs)=\vccdim^*(\mcs^+)$, where $\mcs^+$ is the lifting of
	$\mcs$ to an even set family. Since by Lemma \ref{odd-even},  $\mcs$ and $\mcs^+$ have isomorphic
	1,2-inclusion
	graphs, the density result for  $\mcs$ follows from the density result for $\mcs^+$. This concludes the proof
	of
	Theorem  \ref{density}.
	
	\begin{example}\label{example6}
		As in the case of classical VC-dimension and Theorem \ref{subgraphs_hypercubes},  the inequality from
		Theorem  \ref{density} between the density of 1,2-inclusion graph $G_{1,2}(\mcs)$ and the
		\vccdimt  of $\mcs$ is sharp in the following sense: there exist even set families $\mcs$ such that the
		degeneracy of $G_{1,2}(\mcs)$ equals to
		$\binom{d}{2}$.
		For	example, the sporadic clique $\mcs_2$ has clique VC-dimension 3 (see Examples \ref{example3} and
		remark~\ref{remark5}), degeneracy 3, and density $\frac{3}{2}$). Notice that $G_{1,2}(\mcs_2)$ is the
		halved cube $\frac{1}{2}Q_3$. More generally, let $\mcs_4$  be the even set family  consisting of
		all  even subsets of an $m$-set $X$. Clearly $d:=\vccdim^*(\mcs_4)=|X|=m$ and $\mcs_4$  induces the
		halved cube $\frac{1}{2}Q_m$. We assert that  $\frac{1}{2}Q_m$ has  degeneracy 	
		$\binom{d}{2}$. Indeed, every $S\in\mcs_4$ is incident to $\binom{|X|-|S|}{2}$ supersets of cardinality
		$|S|+2$, to $\binom{|S|}{2}$ subsets of cardinality $|S|-2$, and to $|S|(|X|-|S|)$ sets of cardinality $|S|$.
		Setting $s:=|S|$, we conclude that each set $S$ has degree
		$$
		\frac{(m-s)(m-s-1) + s(s-1)}{2} + s(m-s) = \frac{1}{2}(m^2-m) =\frac{1}{2}(d^2-d)=\binom{d}{2}.
		$$
	\end{example}
	
	\begin{remark} In the following table, for pointed even set families $\mcs_0,\mcs_1,\ldots, \mcs_4$ defined in
		Examples \ref{example1}-\ref{example4} and \ref{example6}, we present their VC-dimension,
		the two clique VC-dimensions, the 2VC-dimension, the degeneracy, and the density.
		\begin{table}[!h]
			\renewcommand\arraystretch{1.3}
			\centering
			\begin{tabular}{|c|c|c|c|c|c|c|c|}
				\hline
				$\mcs$ &$\vcd$ &$\vccdim$ &$\vccdim^*$ & degeneracy &density &2\text{VC-dim}
				\\\hline
				$\mcs_0$	& $1$	& $-$ & $m$	& $m-1$& $\frac{m-1}{2}$ &$0$ 
				\\\hline
				$\mcs_1$	& $1$	& $m$ & $m$	& $m-1$& $\frac{m-1}{2}$ &$2$ 
				\\\hline
				$\mcs_2$	& $2$	&$3$	&$3$	&$3$ &$\frac{3}{2}$ 	&$3$ 
				\\\hline
				$\mcs_3$	& $2$	&$k+2$	&$k+2$	&$k$ &$\frac{k}{2}+o(1)$ 	&$2$ 
				\\\hline
				$\mcs_3\triangle X$	& $2$	&$(m-1)k+1$	&$k+2$	&$k$ &$\frac{k}{2}+o(1)$ 	&$2$ 
				\\\hline
				$\mcs_4$	&$m-1$	&$m$	&$m$	&$\binom{m}{2}$	&$\frac{1}{2}\binom{m}{2}$
				&$m$ 
				\\\hline
			\end{tabular}
			\end{table}
	\end{remark}

	\section{Final discussion}
	In this note, we adapted the shifting techniques to prove that if $\mcs$ is an arbitrary set family and
	$\incgraph(\mcs)=(V,E)$ is the 1,2-inclusion graph of $\mcs$, then $\frac{|E|}{|V|}\le  \binom{d}{2}$,
	where $d:=\vccdim^*(\mcs)$ is the clique-VC-dimension of $\mcs$. The essential ingredients of our proof are
	Proposition \ref{lem_push-down_augmente_densite} (showing that d-shiftings preserve the number of
	vertices and do not decrease the number of edges), Lemma \ref{lem_push-down_baisse_vccdim2} (showing
	that d-shiftings do not increase the clique-VC-dimension), and Proposition \ref{bouquet-density} (bounding
	the density of bouquets of halved cubes, resulting from complete d-shiftings), all established for even set
	families. While Propositions
	\ref{lem_push-down_augmente_densite} and \ref{bouquet-density} are not very sensitive to the chosen
	definition of the clique-VC-dimension (but they require using the definition of 1,2-inclusion graphs as
	the subgraphs of the halved cube
	$\frac{1}{2}Q_m$), Lemma \ref{lem_push-down_baisse_vccdim2} strongly depends on how the
	clique-VC-dimension is defined. For example, this lemma does not hold for the notion of 2VC-dimension of
	\cite{BouTh} discussed in Section \ref{sect_original_method}. Notice also that, differently from the classical
	VC-dimension and similarly to our notion of clique-VC-dimension, 2VC-dimension is not invariant under
	twistings.
	
	In analogy to 2-shattering and 2VC-dimension, we can define the concepts of star-shattering and
	star-VC-dimension,
	which might be useful for finding sharper upper bounds (than those obtained in this
	paper) for density of 1,2-inclusion graphs. Let $Y\subset X$ and $e\notin Y$. We say that a set family $\mcs$
	{\it star-shatters} (or {\it s-shatters}) the pair $(e,Y)$ if for any $y\in Y$ there exists a set $S\in \mcs$ such
	that $S\cap (Y\cup\{ e\})=\{ e,y\}$. The {\it star-VC-dimension} of a pointed set family $\mcs$
	is
	$$
		\vcsdim(\mcs) := \max\{|Y|+1 : Y\subset X\text{ and }\exists e_i\in X\setminus Y\text{ such that } (e_i,Y)
	\text{ is s-shattered by }\mcs\}.
	$$
	The difference with c-shattering is that, in the definition of s-shattering,
	a pair $(e,Y)$ is s-shattered if all $2$-sets of $P(e,Y)$ have non-empty
	fibers, i.e., if $P(e,Y) \subseteq \pi_{Q(e,Y)}(\mcs)$.
	Consequently, any s-shattered pair $(e,Y)$ is c-shattered,
	thus $\vcsdim(\mcs)\le \vccdim(\mcs)$. Since $\vcsdim(\mcs_3\triangle X)=3$ and $G_{1,2}(\mcs_3\triangle
	X)$
	contains a clique of size $k+1$, $\vcsdim(\mcs)$ cannot be used directly to bound the density of 1,2-inclusion
	graphs. We can adapt this notion by taking the maximum over all twistings with respect to sets of $\mcs$:
	the
	{\it star-VC-dimension} $\vcsdim^*(\mcs)$ of an arbitrary set family $\mcs$ is $\max \{ \vcsdim(\mcs\Delta A):
	A\in \mcs\}$\footnote{As noticed by one referee and O. Bousquet, in this form, the star-VC-dimension minus
	one coincides with the notion of {\it star number} that has been studied in the context of active learning
	\cite[Definition 2]{HaYa}.}. Even if $\vcsdim(\mcs)\le \vccdim(\mcs)$ holds for pointed families, as the following examples show,
    there are no relationships between $\vccdim^*(\mcs)$ and $\vcsdim^*(\mcs)$ for even families.

    \begin{example} Let $X=\{1,2,\ldots,2m-1,2m\}$, where $m$ is an arbitrary even integer,  and let
	$\mcs_5:=\{\varnothing \} \cup \{ \{1,2,\ldots,2i-1,2i\}: i=1,\ldots,m\}$. The nonempty sets of $\mcs_5$ can be viewed as
	intervals of even length of $\mathbb N$ with a common origin. The 1,2-inclusion graph of $\mcs_5$ is a
	path of length $m$. For any set $\{1,2,\ldots,2i\}$, the twisted family $\mcs_5^i :=\mcs_5\triangle\{1,2,\ldots,2i\}$ is the
	union of the set families $\mcs':=\{\varnothing, \{2i+1,2i+2\},\ldots,\{2i+1,2i+2,\ldots,2m\}\}$ and $\mcs'':=\{
	\{1,2,\ldots,2i-1,2i\},\ldots,\{2i-1,2i\} \}$. We assert that for any $i=1,\ldots,m$, we have $\vcsdim(\mcs_5^i)\leq3$ and
	$\vccdim(\mcs_5^i)=\max\{i, m-i\}+1$. Indeed, for any element $j\in X$, $\mcs_5^i$ cannot simultaneously s-shatter
	two pairs $\{j,l_1\}$, $\{j,l_2\}$ with $j < l_1 < l_2$ because every set of $\mcs_5^i$ containing $l_2$ also contains
	$l_1$. Analogously, $\mcs_5^i$ cannot s-shatter two pairs $\{j,l_1\}$ and $\{j,l_2\}$ with $l_2<l_1<j$.
	Consequently, if the pair $(j,Y)$ is s-shattered by $\mcs_5^i$, then $|Y|\leq 2$. This shows that
	$\vcsdim^*(\mcs_5)\leq 3$.

	 To see that $\vccdim(\mcs_5^i)=\max\{i, m-i\}+1$, notice that $\mcs'$ c-shatters
	the pair $(2i+1,Y')$ with $Y':=\{2i+2,2i+4,\ldots,2m\}$ and $\mcs''$ c-shatters the pair $(2i,Y'')$ with
   $Y'':=\{1,3,\ldots,2i-1\}$. Since the minimum over all $i=1,\ldots,m$ of $\max\{i, m-i\}+1$ is attained for
   $i=\frac{m}{2}$, we conclude that $\vccdim^*(\mcs_5)=\frac{m}{2}+1$. Therefore $\vcsdim^*(\mcs)$ can be arbitrarily
   smaller than $\vccdim^*(\mcs)$.
	\end{example}
	
    \begin{example}
	Let $X=X_1\dot{\cup} X_2$ with $X_1=\{e_1,\ldots,e_m\}$ and $X_2=\{x_1,\ldots,x_m\}$, and let
	$\mcs_6 := \{\varnothing, \{e_1,x_1\} \} \cup \{ \{e_1,e_i,x_1,x_i\} : 2\leq i\leq m \}$. The 1,2-inclusion graph
	of $\mcs_6$ is a star. One can easily see that $\vcsdim(\mcs_6) = m$. On the other hand, for the twisted family
	$\mcs'_6:=\mcs_6\triangle\{e_1,x_1\} = \{\varnothing\} \cup \{ \{e_i,x_i\} : 1\leq i\leq m \}$, one can check that
    $\vccdim(\mcs_6')= 2$, showing that $\vccdim^*(\mcs_6) = 2$ and $\vcsdim^*(\mcs_6) = m$. Therefore $\vcsdim^*(\mcs)$
    can be arbitrarily larger than $\vccdim^*(\mcs)$.
    \end{example}
	
	Therefore, it is natural to ask whether in Theorem  \ref{density} one can replace $\vccdim^*(\mcs)$ by
	$\vcsdim^*(\mcs)$.
    However, we were not able to decide the status of the following question:
	
	\begin{question} \label{question1}
		Is it true that for any (even) set family $\mcs$ with the 1,2-inclusion graph $G_{1,2}(\mcs)=(V,E)$ and
		star-VC-dimension $d=\vcsdim^*(\mcs)$, we have $\frac{|E|}{|V|}= O(d^2)$?
	\end{question}
	
	The main difficulty here is that a d-shifting may increase the star-VC-dimension, i.e., Lemma
	\ref{lem_push-down_baisse_vccdim2} does no longer hold. The difference between the s-shattering and
	c-shattering is that a $2$-set $\{e,y\}$ with $y\in Y$ can be s-shattered only by a set $S\in \mcs$ which
	belongs to the fiber $F(\{ e,y\})$ (the requirement $Y\cap S=\{ e,y\}$), while  $\{e,y\}$
	can be c-shattered by a set $S$ if $S$ just includes this set (the requirement $\{ e,y\}\subseteq S$). When
	performing a d-shifting $\phiij$ with respect to a pair $\{e_i,e_j\}$ such that $\{e_i,e_j\}\cap\{e,y\}=\varnothing$, a set $S\in
	\mcs$ can be mapped to a set $\phiij(S)$ belonging to the
	fiber $F(\{ e,y\})$. If $\phiij(S)$ is used to c-shatter the $2$-set $\{e,y\}$ by $\phiij(\mcs)$, then $S$
	can be used to shatter $\{e,y\}$ by $\mcs$ (the proof of Lemma \ref{lem_push-down_baisse_vccdim2}).
	However, this is no longer true for s-shattering, because initially $S$ may not necessarily belong to $F(\{
	e,y\})$.
	
	Also we have not found a counterexample to the following question (where the square of the
	clique-VC-dimension or of the star-VC-dimension is replaced by the product of the classical VC-dimension of
	$\mcs$ and the clique number of $G_{1,2}(\mcs)$):
	
	\begin{question} \label{question2}
		Is it true that for any set family $\mcs$ with 1,2-inclusion graph $G_{1,2}(\mcs)=(V,E)$, $d=\vcd(\mcs)$,
		and clique number $\omega=\omega(G_{1,2}(\mcs))$, we have $\frac{|E|}{|V|}=O(d\cdot\omega)$?
	\end{question}
	
	Hypercubes are subgraphs of Johnson graphs, therefore they are 1,2-inclusion graphs. This shows the
	necessity of both parameters (VC-dimension and clique number) in the formulation of
	Question~\ref{question2}. As above, the bottleneck in solving Question~\ref{question2} via shifting is that
	this operation may increase the clique number of 1,2-inclusion graphs.
	
	An alternative approach to Questions \ref{question1} and \ref{question2} is to adapt the original proof of
	Theorem \ref{subgraphs_hypercubes} given in \cite{HaLiWa}. In brief, for a set family $\mcs$ of VC-dimension
	$d$ and an element $e$, let $\mcs_e=\{ S'\subseteq X\setminus \{ e\}: S'=S\cap X \text{ for some } S\in \mcs\}$
	and $\mcs^e=\{ S'\subseteq X\setminus \{ e\}: S' \mbox{ and } S'\cup \{ e\} \mbox{ belong to } \mcs\}.$ Then
	$|\mcs|=|\mcs_e|+|\mcs^e|$, $\vcd(\mcs_e)\le d$, and $\vcd(\mcs^e)\le d-1$ hold. Denote by $G_e$ and $G^e$
	the 1-inclusion graphs of $\mcs_e$ and $\mcs^e$. Then $|E(G_e)|\le d|V(G_e)|=d|\mcs_e|$ and $E(G^e)\le
	(d-1)|V(G^e)|=(d-1)|\mcs^e|$ by induction hypothesis. The proof of the required density inequality follows by
	induction from the equality $|V(G)|=|\mcs|=|\mcs_e|+|\mcs^e|=|V(G_e)|+|V(G^e)|$ and the inequality
	$|E(G)|\leq |E(G_e)|+|E(G^e)|+|V(G^e)|$. Unfortunately, as was the case for shiftings, the clique number of
	$G_{1,2}(\mcs_e)$ may be strictly larger than the clique number of $G_{1,2}(\mcs)$. Also the inequality
	$|E(G)|\leq |E(G_e)|+|E(G^e)|+|V(G^e)|$ is no longer true in this form if instead of 1-inclusion graphs one
	consider 1,2-inclusion graphs.

\medskip\noindent
{\bf Acknowledgements.} We would like to acknowledge the anonymous referees for a careful reading of
the previous version and many useful remarks. We are especially indebted to one of the referees who found a
critical error in the previous proof of Proposition \ref{lem_push-down_augmente_densite}. We would like to
acknowledge another referee and Olivier Bousquet for pointing to us the paper \cite{HaYa}. This work was supported in part
by ANR project DISTANCIA (ANR-17-CE40-0015).
	
\bibliographystyle{amsalpha}

\end{document}